\documentclass[11pt, english]{article}

\usepackage[a4paper,bindingoffset=0.2in,%
            left=1in,right=1in,top=1in,bottom=1in,%
            footskip=.25in]{geometry}
\usepackage{graphicx}%
\usepackage{amsmath}%
\usepackage{amssymb}%
\usepackage{bm}%
\usepackage{paralist}%
\usepackage{bbold}%
\usepackage{mleftright}%
\usepackage{xcolor}%
\usepackage{stmaryrd}%
\usepackage{xspace}%
\usepackage{tablefootnote}%
\usepackage{array}%
\usepackage{titlesec}%
\usepackage[ulem=normalem]{changes}

\usepackage[thmmarks]{ntheorem}%
\theoremseparator{.}%

\usepackage{titlesec}%
\titlelabel{\thetitle. }%

\usepackage{hyperref}%
\hypersetup{%
   breaklinks,%
   colorlinks=true,%
   linkcolor=[rgb]{0.45,0.0,0.0},%
   citecolor=[rgb]{0,0,0.45}
}

\numberwithin{figure}{section}%
\numberwithin{table}{section}%
\numberwithin{equation}{section}%

\providecommand{\IntRange}[1]{\left\llbracket #1 \right\rrbracket}

\newcommand{\IY}[2]{\IntRange{ #1: #2 }}%
\newcommand{\SumX}[1]{\bm{\sum}_{\leq u} \pth {#1}}
\newcommand{\SSumX}[1]{\bm{\sum} \pth {#1}}

\newcommand{\SumT}[1]{\bm{\sum}^{\leq \alpha} \left[{#1} \right]}

\newcommand{\SumSetX}[1]{\Sigma_{#1}}%
\newcommand{\cardin}[1]{\left| {#1} \right|}
\newcommand{\U}[1]{U({#1})}

\newcommand{\seg}[2]{{#1}{\IntRange{{#2}}}}
\newcommand{\segment}{segment}
\newcommand{\segments}{segments}
\newcommand{\length}{length}

\newcommand{\divides}{\mid}
\newcommand{\notdivides}{\nmid}

\newcommand{\polylog}{\operatorname{polylog}}

\definecolor{blue25}{rgb}{0, 0, 11}
\definecolor{default1}{rgb}{0.45,0.0,0.0}
\newcommand{\emphic}[2]{%
   \textcolor{default1}{%
      \textbf{\emph{#1}}}%
   \index{#2}}

\newcommand{\emphi}[1]{\emphic{#1}{#1}}

\newcommand{\pth}[2][\!]{\mleft({#2}\mright)}%

\newcommand{\FFT}{\textsf{F{F}T}\xspace}%
\newcommand{\mY}[2]{\mathbb{1}_{#1}\pth{#2}}

\newcommand{\setOf}[1]{\mathrm{set}\pth{#1}}%
\newcommand{\cardX}[1]{\mathrm{card}\pth{#1}}%

\newcommand{\Set}[2]{\left\{ #1 \;\middle\vert\; #2 \right\}}
\newcommand{\brc}[1]{\left\{ {#1} \right\}}

\renewcommand{\th}{th\xspace}

\newcommand{\myqedsymbol}{\rule{2mm}{2mm}}

\newcommand{\floor}[1]{\left\lfloor {#1} \right\rfloor}

\newcommand{\HLinkShort}[2]{\hyperref[#2]{#1\ref*{#2}}}
\newcommand{\HLink}[2]{\hyperref[#2]{#1~\ref*{#2}}}
\newcommand{\HLinkPage}[2]{\hyperref[#2]{#1~\ref*{#2}%
      $_\text{p\pageref{#2}}$}}
\newcommand{\HLinkPageOnly}[1]{\hyperref[#1]{Page~\refpage*{#1}%
      $_\text{p\pageref{#1}}$}}

\newcommand{\HLinkSuffix}[3]{\hyperref[#2]{#1\ref*{#2}{#3}}}
\newcommand{\HLinkPageSuffix}[3]{\hyperref[#2]{#1\ref*{#2}%
      #3$_\text{p\pageref{#2}}$}}

\newcommand{\tablab}[1]{\label{table:#1}}
\newcommand{\tabref}[1]{\HLink{Table}{table:#1}}%

\newcommand{\lemlab}[1]{\label{lemma:#1}}
\newcommand{\lemref}[1]{\HLink{Lemma}{lemma:#1}}%

\newcommand{\obslab}[1]{\label{observation:#1}}
\newcommand{\obsref}[1]{\HLink{Observation}{observation:#1}}%

\newcommand{\thmlab}[1]{\label{theorem:#1}}
\newcommand{\thmref}[1]{\HLink{Theorem}{theorem:#1}}%

\newcommand{\N}{\mathbb{N}}
\newcommand{\Z}{\mathbb{Z}}
\newcommand{\R}{\mathbb{R}}

\theoremstyle{plain}%
\newtheorem{theorem}{Theorem}[section]

\newtheorem{lemma}[theorem]{Lemma}

\newtheorem{corollary}[theorem]{Corollary}
 
\newtheorem{observation}[theorem]{Observation}

\theoremstyle{plain}%
\theoremheaderfont{\sf}%
\theorembodyfont{\upshape}%
\newtheorem*{remark:unnumbered}[theorem]{Remark}%
\newtheorem{remark}[theorem]{Remark}%

\theoremheaderfont{\em}%
\theorembodyfont{\upshape}%
\theoremstyle{nonumberplain}%
\theoremseparator{}%
\theoremsymbol{\myqedsymbol}%
\newtheorem{proof}{Proof:}%

\usepackage{comment}
 
\title{\Large A Faster Pseudopolynomial Time Algorithm for Subset Sum}
\author{
    Konstantinos Koiliaris\thanks{Department of Computer Science, University of Illinois, Urbana - Champaign. E-mail: 
    \texttt{koiliar2@illinois.edu}}
    \qquad\quad
    Chao Xu\thanks{Department of Computer Science, University of Illinois, Urbana - Champaign. E-mail: \texttt{chaoxu3@illinois.edu}}
    \date{}\\
}

\begin{document}

\maketitle
\begin{abstract}\small
Given a multiset $S$ of $n$ positive integers and a target integer $t$, the subset sum problem is to decide if there is a subset of $S$ that sums up to $t$. We present a new divide-and-conquer algorithm that computes \emph{all} the realizable subset sums up to an integer $u$ in
 $\widetilde{O}(\min\{\sqrt{n}u,u^{4/3},\sigma\})$,
 where $\sigma$ is the sum of all elements in $S$ and $\widetilde{O}$ hides polylogarithmic factors. This result improves upon the standard dynamic programming algorithm that runs in $O(nu)$ time. To the best of our knowledge, the new algorithm is the fastest general deterministic algorithm for this problem. We also present a modified algorithm for finite cyclic groups, which computes all the realizable subset sums within the group in
  $\widetilde{O}(\min\{\sqrt{n}m,m^{5/4}\})$ 
 time, where $m$ is the order of the group. \\

\end{abstract}
\section{Introduction}
Given a multiset $S$ of $n$ positive integers and an integer target value $t$, the \emph{subset sum problem} is to decide if there is a subset of $S$ that sums to $t$. The subset sum problem is related to the knapsack problem \cite{Dantzig1957} and it is one of Karp's original NP-complete problems \cite{Karp1972}. The subset sum is a fundamental problem used as a standard example of a problem that can be solved in weakly polynomial time in many undergraduate algorithms/complexity classes. As a weakly NP-complete problem, there is a standard pseudopolynomial time algorithm using a dynamic programming, due to Bellman, that solves it in $O(nt)$ time \cite{Bellman1956} (see also \cite[Chapter 34.5]{cormen2014introduction}). The current state of the art has since been improved by a $\log t$ factor using a bit-packing technique \cite{Pisinger2003}. There is extensive work on the subset sum problem, see \tabref{tab:results} for a summary of previous deterministic pseudopolynomial time results \cite{Bellman1956,Pisinger19991,Faaland1973,Pferschy,NET:NET5,Pisinger2003,Lokshtanov2010,Serang2014,Serang2015}. 

\begin{table*}[t]
\def\arraystretch{1.5}
\centering
\begin{tabular}{ m{2.5cm} l l m{5.5cm}}
    \hline
    \textbf{Result} & \textbf{Time} & \textbf{Space} & \textbf{Comments} \\
    \hline\hline
    Bellman\:\cite{Bellman1956} & $O(nt)$ & $O(t)$ & original DP solution \\
    \hline
    Pisinger\:\cite{Pisinger2003} & $O\!\left(\frac{nt}{\log t}\right)$ & $O\!\left(\frac{t}{\log t}\right)$ & RAM model implementation of Bellman\\
    \hline
    Pisinger\:\cite{Pisinger19991} & $O(n\max{S})$ & $O(t)$ & fast if small $\max S$ \\
    \hline
    Faaland\:\cite{Faaland1973},\qquad Pferschy\:\cite{Pferschy} & $O(n't)$ & $O(t)$ & fast for small $n'$ \\
    \hline
    Klinz et al. \cite{NET:NET5} & $O(\sigma^{3/2})$ & $O(t)$ & fast for small $\sigma$, obtainable from above because $n'=O\left(\sqrt{\sigma}\right)$ \\
    \hline
    Eppstein\:\cite{Eppstein1997375}, Serang\:\cite{Serang2014,Serang2015} & $\widetilde{O}\!\left(n\max{S}\right)$ & $O(t\log t)$ & data structure \\
    \hline
    Lokshtanov et al. \cite{Lokshtanov2010} & $\widetilde{O}(n^3t)$ & $\widetilde{O}(n^2)$ & polynomial space \\
    \hline
    \textbf{current work} & \begin{minipage}[t][1.1cm][t]{4cm} $\widetilde{O}\!\left(\min\left\{\sqrt{n'}\,t,t^{4/3},\sigma\right\}\right)$ \thmref{distincttheorem} \end{minipage} & $O(t)$ & see Section \ref{contrib} \\
    \hline
\end{tabular}
\caption{\small Summary of deterministic pseudopolynomial time results on the subset sum problem. The input is a target number $t$ and a multiset $S$ of $n$ numbers, with $n'$ distinct values up to $t$, and $\sigma$ denotes the sum of all elements in $S$.}
\tablab{tab:results}
\end{table*}

Moreover, there are results on subset sum that depend on properties of the input, as well as data structures that maintain subset sums under standard operations. In particular, when the maximum value of any integer in $S$ is relatively small compared to the number of elements $n$, and the target value $t$ lies close to one-half the total sum of the elements, then one can solve the subset sum problem in almost linear time \cite{Galil}. This was improved by Chaimovich \cite{Chaimovich1999}. Furthermore, Eppstein described a data structure which efficiently maintains all subset sums up to a given value $u$, under insertion and deletion of elements, in $O(u\log u \log n)$ time per update, which can be accelerated to $O(u\log u)$ when additional information about future updates is known \cite{Eppstein1997375}. The probabilistic convolution tree, by Serang \cite{Serang2014,Serang2015}, is also able to solve the subset sum problem in $\widetilde{O}(n\max(S))$ time, where $\widetilde{O}$ hides polylogarithmic factors.

If randomization is allowed, more algorithms are possible. In particular, Bringmann showed a randomized algorithm that solves the problem in $\widetilde{O}(nt)$ time, using only $\widetilde{O}(n\log t)$ space under the Extended Riemann Hypothesis \cite{Bringmann16}. Bringmann also provided a randomized near linear time algorithm $\widetilde{O}(n+t)$ -- it remains open whether this algorithm can be derandomized.

Finally, it is unlikely that any subset sum algorithm runs in time $O(t^{1-\epsilon}\,n^c)$, for any constant $c$ and $\epsilon>0$, as such an algorithm would imply that there are faster algorithms for a wide variety of problems including set cover \cite{Bringmann16,Cygan:2012:PHC:2353734.2354303}. 

%
\subsection{Applications of the subset sum problem.}
The subset sum problem has a variety of applications including: power indices \cite{Uno2012}, scheduling \cite{NAV:NAV20202,4358069,gueret1999}, set-based queries in databases \cite{TranCW11}, breaking precise query protocols \cite{Dautrich:2013:CPP:2452376.2452397} and various other graph problems with cardinality constraints \cite{JGT:JGT3,Diaz2006,Cai2006,guruswami,NET:NET5,Eppstein1997375} (for a survey of further applications see \cite{kellerer2004knapsack}).

A faster pseudopolynomial time algorithm for the subset sum would imply faster \emph{polynomial time} algorithms for a number of problems. The \emph{bottleneck graph partition} problem on weighted graphs is one such example. It asks to split the vertices of the graph into two equal-sized sets such that the value of the bottleneck (maximum-weight) edge, over all edges across the cut, is minimized. The impact of our results on this problem and other selected applications is highlighted in Section \ref{applications}. 

\subsection{Our contributions.}
\label{contrib}
The new results are summarized in \tabref{tab:contribution} --  we consider the following \emph{all subset sums} problem: Given a multiset $S$ of $n$ elements, with $n'$ distinct values, with $\sigma$ being the total sum of its elements, compute \emph{all} the realizable subset sums up to a prespecified integer $u$. Computing all subset sums for some $u \geq t$ also answers the standard subset sum problem with target value $t$.

\begin{table}[t]
\def\arraystretch{1.5}
\centering
\begin{tabular}{ c c c }
  \hline
  \textbf{Parameters} & \textbf{Previous best} & \textbf{Current work} \\
  \hline\hline
    $n$ and $t$ & $O(nt/\log t)$ & $\widetilde{O}\!\left(\min\left\{ \sqrt{n}t,\:t^{4/3} \right\} \right)$\\
  \hline
    $n'$ and $t$ & $O(n't)$ & $\widetilde{O}\!\left(\min\left\{ \sqrt{n'}t,\:t^{4/3} \right\}\right)$\\
  \hline
    $\sigma$ & $O(\sigma^{3/2})$ & $\widetilde{O}(\sigma)$\\
  \hline
\end{tabular}
\caption{\small Our contribution on the subset sum problem compared to the previous best known results. The input $S$ is a multiset of $n$ numbers with $n'$ distinct values, $\sigma$ denotes the sum of all elements in $S$ and $t$ is the target number.} 
\tablab{tab:contribution}
\end{table}

Our main contribution is a new algorithm for computing the all subset sums problem in $\widetilde{O}\!\left( \min\{\sqrt{n}u, u^{4/3}, \sigma\} \right)$ time. The new algorithm improves over all previous work (see \tabref{tab:contribution}). To the best of our knowledge, it is the fastest general deterministic pseudopolynomial time algorithm for the all subset sum problem, and consequently, for the subset sum problem.

Our second contribution is an algorithm that solves the \emph{all subset sums} problem modulo $m$, in $O\!\left(\min \{\sqrt{n}m,m^{5/4}\} \log^2 m\right)$ time. Though the time bound is superficially similar to the first algorithm, this algorithm uses a significantly different approach.

Both algorithms can be augmented to return the solution; i.e., the subset summing up to each number, with a polylogarithmic slowdown (see Section \ref{sec:recover} for details).


%
\subsection{Sketch of techniques.}
The straightforward divide-and-conquer algorithm for solving the subset sum problem \cite{Horowitz:1974:CPA:321812.321823}, partitions the set of numbers into two sets, recursively computes their subset sums and combines them together using \FFT  \cite{Eppstein1997375,Serang2014,Serang2015} (Fast Fourier Transform \cite[Chapter 30]{cormen2014introduction}). This algorithm has a running time of $O(\sigma \log \sigma \log n)$. 



\paragraph{Sketch of the first algorithm (on integers).} Our main new idea is to improve the ``conquer'' step by taking advantage of the structure of the sets. In particular, if $S$ and $T$ lie in a short interval, then one can combine their subset sums quickly, due to their special structure. On the other hand, if $S$ and $T$ lie in a long interval, but the smallest number of the interval is large, then one can combine their subset sums quickly by ignoring most of the sums that exceed the upper bound.

The new algorithm works by first partitioning the interval $\IY{0}{u}$ into a logarithmic number of exponentially long intervals. Then computes these partial sums recursively and combines them together by aggressively deploying the above observation.

\paragraph{Sketch of the second algorithm (modulo $m$).} Assume $m$ is a prime number. Using known results from number theory, we show that for any $\ell$ one can partition the input set into $\widetilde{O}(|S|/\ell)$ subsets, such that every such subset is contained in an arithmetic progression of the form $x, 2x, \ldots, \ell x$. The subset sums for such a set can be quickly computed by dividing and later multiplying the numbers by $\ell$. Then combine all these subset sums to get the result.

Sadly, $m$ is not always prime. Fortunately, all the numbers that are relative prime to $m$ can be handled in the same way as above. For the remaining numbers we use a recursive partition classifying each number, in a sieve-like process, according to which prime factors it shares with $m$. In the resulting subproblems all the numbers are coprime to the moduli used, and as such the above algorithm can be used. Finally, the algorithm combines the subset sums of the subproblems.

\paragraph{Paper organization.} Section \ref{sec:section2} covers the algorithm for positive integers. Section \ref{sec:cycl} describes the algorithm for the case of modulo $m$. 
Section \ref{sec:recover} shows how we can recover the subsets summing to each set, and Section \ref{applications} presents the impact of the results on selected applications of the problem. 

\section{The algorithm for integers}
\label{sec:section2}

\subsection{Notations.}

Let $\IY{x}{y} = \brc{x,x+1, \ldots, y}$ denote the set of integers in the interval $[x,y]$. Similarly, $\IntRange{x} = \IY{1}{x}$. For two sets $X$ and $Y$, we denote by $X \oplus Y$ the set $\Set{x+y}{x \in X \text{ and }y \in Y}$. If $X$ and $Y$ are sets of points in the plane, $X\oplus Y$ is the set $\{(x_1+y_1, x_2 + y_2) \:|\: x_1, x_2 \in X \text{ and } y_1, y_2 \in Y\}$.

For an element $s$ in a multiset $S$, its \emphi{multiplicity} in $S$ is denoted by $\mY{S}{s}$. We denote by $\setOf{S}$ the set of \emph{distinct elements} appearing in the multiset $S$. The \emphi{size} of a multiset $S$ is the number of distinct elements in $S$ (i.e., $\cardin{\setOf{S}}$). The \emphi{cardinality} of $S$, is $\cardX{S} = \sum_{s \in S} \mY{S}{s}$. We denote that a multiset $S$ has all its elements in the interval $\IY{x}{y}$ by $S \subseteq \IY{x}{y}$.

For a multiset $S$ of integers, let $\SumSetX{S} = \sum_{s \in S} \mY{S}{s}\cdot s$ denote the total \emphi{sum} of the elements of $S$. The set of \emphi{all subset sums} is denoted by \[\SSumX{S} = \Set{\SumSetX{T}}{T \subseteq S}\:.\] The pair of the set of all subset sums using sets of size at most $\alpha$ along with their associated cardinality is denoted by $\SumT{S} = \left\{ \bigl(\SumSetX{T},\,|T|\bigr) \:\middle|\: T\subseteq S,\, |T|\leq \alpha \right\}$. The set of all subset sums of a set $S$ up to a number $u$ is denoted by $\SumX{S} = \SSumX{S} \cap \IY{0}{u}$.

\subsection{From multisets to sets.}
\label{frommultisetstosets}

Here, we show that the case where the input is a multiset can be reduced to the case of a set. The reduction idea is somewhat standard (see \cite[Section 7.1.1]{kellerer2004knapsack}), and first appeared in \cite{Lawler1979}. We present it here for completeness.

\begin{lemma}
Given a multiset $S$ of integers, and a number $s \in S$, with $\mY{S}{s} \geq 3$. Consider the multiset $S'$ resulting from removing two copies of $s$ from $S$, and adding the number $2s$ to it. Then, $\SumX{S}= \SumX{S'}$.  Observe that $\cardX{S'} = \cardX{S} -1$.
\end{lemma}

\begin{proof}
    Consider any multiset $T \subseteq S$. If $T$ contains two or more
    copies of $s$, then replace two copies by a single
    copy of $2s$. The resulting subset is $T' \subseteq S'$, and
    $\SumSetX{T} = \SumSetX{T'}$, establishing the claim.
\end{proof}

\begin{lemma}
    \lemlab{sparsify}%
    Given a multiset $S$ of integers in $\IntRange{u}$ of cardinality
    $n$ with $n'$ unique values, one can compute, in $O(n' \log^2 u)$ time, a multiset $T$, such
    that:
    \begin{inparaenum}[(i)]
        \item $\SumX{S}= \SumX{T}$,\\
        \item $\cardX{T} \leq \cardX{S}$, \\
        \item $\cardX{T} = O(n'\log u)$, and \\
        \item no element in $T$ has multiplicity exceeding two.
    \end{inparaenum}
\end{lemma}
\begin{proof}
    Copy the elements of $S$ into a working multiset $X$. Maintain the elements of $\setOf{X}$ in a heap $D$, and let $T$ initially be the empty set. In each iteration, extract the minimum element $x$ from the heap $D$. If $x>u$, we stop.

    If $\mY{X}{x} \leq 2$, then delete $x$ from $X$, and add $x$, with its appropriate multiplicity, to the output multiset $T$, and continue to the next iteration.

    If $\mY{X}{x} > 2$, then delete $x$ from $X$, add $x$ to the output set $T$ (with multiplicity one), insert the number $2x$ into $X$ with multiplicity $m' = \floor{(\mY{X}{x} -1) /2}$, (updating also the heap $D$ -- by adding $2x$ if it is not already in it), and set
    $\mY{X}{x} \leftarrow \mY{X}{x} - 2 m'$. The algorithm now continues to the next iteration.

    At any point in time, we have that $\SumX{S}= \SumX{X \cup T}$, and every iteration takes $O( \log u)$ time, and and as such overall, the running time is $O(\cardX{T} \log u)$, as each iteration increases $\cardX{T}$ by at most two. Finally, notice that every element in $T$ is of the form $2^i x, x \in S$ for some $i$, where $i \leq \log n$, and thus $\cardX{T} = O(n' \log u)$.
\end{proof}

Note that the following lemma refers to sets.

\begin{lemma}
    \lemlab{o:plus}%
    Given two \emph{sets} $S, T \subseteq \IY{0}{u}$, one can compute $S \oplus T$ in $O( u \log u )$ time. 
\end{lemma}
\begin{proof}
Let $f_S(x) = \sum_{i \in S} x^i$ be the \emph{characteristic polynomial} of $S$. Construct, in a similar fashion, the polynomial $f_T$ and let $g = f_S * f_T$. Observe that the coefficient of $x^i$ in $g$ is greater than 0 if and only if $i \in S \oplus T$. As such, using \FFT, one can compute the polynomial $g$ in $O(u \log u)$ time, and extract $S \oplus T$ from it.
\end{proof}

\begin{observation} 
    If $P$ and $Q$ form a partition of multiset $S$, then $\SSumX{S} = \SSumX{P} \oplus \SSumX{Q}$.
\end{observation}

Combining all of the above together, we can now state the following lemma which simplifies the upcoming analysis.

\begin{lemma}
    \lemlab{onlysets}
    Given an algorithm that computes $\SumX{S}$ in $\mathrm{T}(n,u)=\Omega(u\log^2 u)$ time, for any \emph{set} $S\subseteq \IntRange{u}$ with $n$ elements, then one can compute $\SumX{S'}$ for any \emph{multiset} $S'\subseteq \IntRange{u}$, with $n'$ distinct elements, in $O\bigl(\mathrm{T}(n'\log u,u)\bigr)$ time.
\end{lemma}
\begin{proof}
   First, from $S$, compute the multiset $T$ as described in \lemref{sparsify}, in $O(u\log^2 u)$ time. As every element in $T$ appears at most twice, partition it into two \emph{sets} $P$ and $Q$. Then $\SumX{T} = \left(\SumX{P}\oplus \SumX{Q}\right) \cap \IY{0}{u}$, which is computed using \lemref{o:plus}, in $O(u\log u)$ time. This reduces all subset sums for multisets of $n'$ distinct elements to two instances of all subset sums for sets of size $O(n'\log u)$.~
\end{proof}
\subsection{The input is a set of positive integers.}


In the previous section it was shown that there is little loss in generality and running time if the input is restricted to sets instead of multisets. For simplicity of exposition, we assume the input is a \emph{set} from here on.
 
Here, we present the main algorithm: At a high level it uses a geometric partitioning on the input range $\IY{0}{u}$ to split the numbers into groups of exponentially long intervals. Each of these groups is then processed separately abusing their interval range that bounds the cardinality of the sets from that group.

\begin{observation}
    \obslab{rec}
    Let $g$ be a positive, superadditive (i.e. $g(x+y) \geq g(x)+g(y), \forall x,y$) function. For a function $f(n,m)$ satisfying
    \[f(n,m) = \max_{m_1+m_2 = m}\left\{ f\!\left(\frac{n}{2},m_1\right) + f\!\left(\frac{n}{2},m_2\right) + g(m) \right\}\:,\]
    we have that $f(n,m) = O\left(g(m)\log n\right)$.
\end{observation}

\begin{theorem}
    \thmlab{stupid2}%
    Given a set of positive integers $S$ with total sum $\sigma$, one can compute the set of all subset sums $\SSumX{S}$ in $O(\sigma \log \sigma \log n)$ time.
\end{theorem}
\begin{proof}
Partition $S$ into two sets $L,R$ of (roughly) equal cardinality, and compute recursively $L' = \SSumX{L}$ and $R' = \SSumX{R}$. Next, compute $\SSumX{S} = L' \oplus R'$ using \lemref{o:plus}. The recurrence for the running time is $f(n,\sigma) = \max_{\sigma_1+\sigma_2=\sigma}\{ f(n/2,\sigma_1) + f(n/2,\sigma_2) + O(\sigma \log \sigma)$\}, and the solution to this recurrence, by \obsref{rec}, is $O(\sigma \log \sigma \log n)$.
\end{proof}

\begin{remark} The standard divide-and-conquer algorithm of \thmref{stupid2} was already known in \cite{Serang2015,Eppstein1997375}, here we showed a better analysis. Note, that the basic divide-and-conquer algorithm without the \FFT addition was known much earlier \cite{Horowitz:1974:CPA:321812.321823}.
\end{remark}

\begin{lemma}[\cite{Serang2015,Eppstein1997375}]
    \lemlab{stupid}%
    Given a set $S \subseteq \IntRange{\Delta}$ of size $n$, one can compute the set $\SSumX{S}$ in $O\bigl( n \Delta \log (n\Delta) \log n\bigr)$ time.
\end{lemma}
\begin{proof}
Observe that $\SumSetX{S} \leq \Delta n$ and apply \thmref{stupid2}.~
\end{proof}

\begin{lemma}
    \lemlab{o:plus2}%
    Given two sets of points $S, T \subseteq \IY{0}{u} \times \IY{0}{v}$, one can compute $S \oplus T$ in $O\bigl( uv \log (uv)\bigr)$ time. 
\end{lemma}
\begin{proof}    
Let $f_S(x,y) = \sum_{(i,j) \in S} x^iy^j$ be the characteristic polynomial of $S$. Construct, similarly, the polynomial $f_T$, and let $g = f_S * f_T$. Note that the coefficient of $x^iy^j$ is greater than $0$ if and only if $(i,j)\in S \oplus T$. One can compute the polynomial $g$ by a straightforward reduction to regular \FFT (see multidimensional \FFT \cite[Chapter 12.8]{Blahut:1985:FAD:537283}), in $O( uv \log uv)$ time, and extract $S \oplus T$ from it.
\end{proof}

\begin{lemma}
    \lemlab{s:m:3}%
    Given two disjoint sets $B, C \subseteq \IY{x}{x+\ell}$ and $\SumT{B}$, $\SumT{C}$, one can compute $\SumT{B\cup C}$ in $O\left(\ell \alpha^2 \log (\ell \alpha) \right)$ time.
\end{lemma}
\begin{proof}
Consider the function $f\bigl((i,j)\bigr) = (i - xj,j)$. Let $X=f\left(\SumT{B}\right)$ and $Y= f\left(\SumT{C}\right)$. If $(i,j)\in \SumT{B}\cup \SumT{C}$, then $i=jx+y$ for $y\in \IY{0}{\ell j}$. Hence $X,Y\subseteq \IY{0}{\ell \alpha}\times \IY{0}{\alpha}$.
     
Computing $X\oplus Y$ using the algorithm of \lemref{o:plus2} can be done in $O\left(\ell \alpha^2 \log (\ell \alpha) \right)$ time. Let $Z=(X\oplus Y)\cap (\IY{0}{\ell \alpha} \times \IY{0}{\alpha})$. The set $\SumT{B\cup C}$ is then precisely $f^{-1}(Z)$. Projecting $Z$ back takes an additional $O\left(\ell \alpha^2 \log (\ell \alpha) \right)$ time. 
 \end{proof}

\begin{lemma}
    \lemlab{s:m:2}%
    Given a set $S \subseteq \IY{x}{x+\ell}$ of size $n$, computing the set $\SumT{S}$ takes $O\bigl(\ell \alpha^2 \log (\ell \alpha) \log n\bigr)$ time.
\end{lemma}
\begin{proof}
Compute the median of $S$, denoted by $\delta$, in linear time. Next, partition $S$ into two sets $L = S \cap \IntRange{\delta}$ and $R=S \cap \IY{\delta+1}{x+\ell}$. Compute recursively $L' = \SumT{L}$ and $R' = \SumT{R}$, and combine them into $\SumT{L\cup R}$ using \lemref{s:m:3}. The recurrence for the running time is:
\[ f(n,\ell) = \!\!\!\max_{\ell_1+\ell_2=\ell} \!\left\{\!f\!\left(\frac{n}{2},\ell_1\right)+f\!\left(\frac{n}{2},\ell_2\right) + O\!\left(\ell \alpha^2 \log (\ell \alpha)\right)\!\right\}\:, \]
which takes $O\bigl( \ell \alpha^2 \log (\ell \alpha) \log n\bigr)$ time, by \obsref{rec}. 
\end{proof}

\begin{lemma}
    \lemlab{s:m}%
    Given a set $S \subseteq \IY{x}{x+\ell}$ of size $n$, computing the set $\SumX{S}$ takes $O\left((u/x)^2 \ell \log (\ell u/x) \log n\right)$ time.
\end{lemma}
\begin{proof}
Apply \lemref{s:m:2} by setting $\alpha=\lfloor u/x \rfloor$ to get $\SumT{S}$. Projecting down by ignoring the last coordinate and then intersecting with $\IY{0}{u}$ gives the set $\SumX{S}$.
\end{proof}

\begin{lemma}
    \lemlab{lem:partition}
    Given a set $S \subseteq \IntRange{u}$ of size $n$ and a parameter $r_0\geq 1$, partition $S$ as follows:
    \begin{compactitem}
        \item $S_0 = S \cap \IntRange{r_0}$, and
        \item for $i >0$, $S_i = S \cap \IY{r_{i-1}+1}{r_i}$, where $r_i= \floor{2^i r_0}$.
    \end{compactitem}
    The resulting partition is composed of $\nu = O(\log u)$ sets $S_0, S_1, \ldots, S_\nu$ and can be computed in $O(n \log n)$ time.
\end{lemma}
\begin{proof}
Sort the numbers in $S$, and throw them into the sets, in the obvious fashion. As for the number of sets, observe that $2^i r_0 > u$ when $i>\log u$. As such, after $\log n$ sets, $r_\nu > u$.
\end{proof}

\begin{lemma}
    \lemlab{generate:intervals}%
    Given a set $S \subseteq \IntRange{u}$ of size $n$. For $i=0, \ldots, \nu =O(\log u)$, let $S_i$ be the $i$\th set in the above partition and let $|S_i|=n_i$. One can compute $\SumX{S_i}$, for all $i$, in overall
    $O\left( (u^2/r_0 + \min\{r_0,n\} r_0) \log^2 u \right)$ time.
\end{lemma}
\begin{proof}
Because $S\subseteq \IntRange{u}$, $n=O(u)$. If $i=0$, then $S_0 \subseteq \IntRange{r_0}$, and one can compute $\SumX{S_0}$, in $O( n_0 r_0 \log (n_0 r_0) \log n_0)$ time, using \lemref{stupid}. Since $n_0\leq r_0$ \ and $n_0\leq n$, this simplifies to $O\left(\min\{n,r_0\} r_0 \log^2 u\right)$.

For $i > 0$, the sets $S_i$ contain numbers at least as large as $r_{i-1}$. Moreover, each set $S_i$ is contained in an interval of length $\ell_i = r_i-r_{i-1} = r_{i-1}$. Now, using \lemref{s:m}, one can compute $\SumX{S_i}$ in
\begin{math}
    O\bigl((u/r_{i-1})^2 \ell_i \log (\ell_i u/r_{i-1}) \log n_i\bigr)%
    =%
    O \pth{ \frac{u^2}{r_{i-1}} \log^2 u}
\end{math}
time. Summing this bound, for $i=1,\ldots, \nu$, results in $O\left(\frac{u^2}{r_0} \log^2 u\right)$ running time.
\end{proof}

\begin{theorem}
    \thmlab{theorem:main}
    Let $S \subseteq \IntRange{u}$ be a set of $n$ elements. Computing the set of all subset sums $\SumX{S}$ takes $O\left(\min\{\sqrt{n}u,u^{4/3}\}\log^2 u\right)$ time.
\end{theorem}
\begin{proof}
Assuming the partition of \lemref{lem:partition}, compute the subset sums $T_i = \SumX{S_i}$, for $i=0,\ldots, \nu$. Let $P_1 = T_1$, and let $P_i = (P_{i-1} \oplus T_i) \cap \IntRange{u}$. Each $P_i$ can be computed using the algorithm of \lemref{o:plus}. Do this for $i=1, \ldots, \nu$, and observe that the running time to compute $P_\nu$, given all $T_i$, is $O( \nu ( u \log u )) = O(u \log^2 u) $.

Finally, for all $i=1,\ldots, \nu$ calculating the $T_i$'s:
\begin{compactitem}
    \item By setting $r_0$ equal to $u^{2/3}$ and using \lemref{generate:intervals} takes $O\left(u^{4/3}\log^2 u\right)$. 
    \item By setting $r_0$ equal to $\frac{u}{\sqrt{n}}$ and using 
\lemref{generate:intervals} takes $O\left(\sqrt{n}u\log^2 u\right)$.
\end{compactitem}
Taking the minimum of these two, proves the theorem.~
\end{proof}

Putting together \thmref{stupid2}, \thmref{theorem:main} and \lemref{onlysets}, results in the following when the input is a \emph{multiset}.

\begin{theorem}[Main theorem]
    \thmlab{distincttheorem}
    Let $S \subseteq \IntRange{u}$ be a \emph{multiset} of $n'$ distinct elements, with total sum $\sigma$, computing the set of all subset sums $\SumX{S}$ takes 
    \[ O\!\left(\min\left\{\sqrt{n'}\,u\log^{\frac{5}{2}} u ,u^{\frac{4}{3}}\log^2 u, \sigma\log \sigma \log \left(n'\log u\right)\right\} \right) \]
    time.
\end{theorem}


%
%
\section{Subset sums for finite cyclic groups}
\label{sec:cycl}

In this section, we demonstrate the robustness of the idea underlying the algorithm of Section \ref{sec:section2} by showing how to extend it to work for finite cyclic groups. The challenge is that the previous algorithm throws away many sums that fall outside of $\IntRange{u}$ during its execution, but this can no longer be done for finite cyclic groups, since these sums stay in the group and as such must be accounted for.

\subsection{Notations.}
For any positive integer $m$, the set of integers \emph{modulo} $m$ with the operation of addition forms a finite \emphi{cyclic group}, the group $\Z_m=\{0,1,\ldots, m-1\}$ of \emph{order} $m$. Every finite cyclic group of order $m$ is isomorphic to the group $\Z_m$ (as such it is sufficient for our purposes to work with $\Z_m$). Let $\U{\Z_m} = \{x\in \Z_m \:|\: \gcd(x,m) = 1\}$ be the \emphi{set of units} of $\Z_m$, and let \emph{Euler's totient} function $\varphi(m)=|\U{\Z_m}|$ be the \emphi{number of units} of $\Z_m$. We remind the reader that two integers $\alpha$ and $\beta$ such that $\gcd(\alpha, \beta)=1$ are  \emph{coprime} (or relatively prime). The set 
\[\seg{x}{\ell}=\bigr\{x,2x,\ldots,\ell x\bigl\}\]
is a finite arithmetic progression, henceforth referred to as a \emphi{\segment} of \emphi{\length} $\left|\seg{x}{\ell}\right|=\ell$. Finally, let $S/x = \{s/x \:|\:s\in S \mbox{ and } x\divides s \}$ and $S\%x = \{s \in S \:|\: x\notdivides s \}$, where $x\divides s$ and $x\notdivides s$ denote that ``$s$ divides $q$'' and ``$s$ does not divide $q$'', respectively. For an integer $x$, let $\sigma_0(x)$ denote the \emphi{number of divisors} of $x$ and $\sigma_1(x)$ the \emphi{sum of its divisors}. 

\subsection{Subset sums and segments.}

\begin{lemma}
    \lemlab{stupidcyclic}%
    For a set $S \subseteq \Z_m$ of size $n$, such that $S\subseteq \seg{x}{\ell}$, the set $\SSumX{S}$ can be computed in $O\left(n\ell \log (n\ell) \log n\right)$ time.
\end{lemma}
\begin{proof}
All elements of $\seg{x}{\ell}$ are multiplicities of $x$, and thus $S' := S/x \subseteq \IntRange{\ell}$ is a well defined set of integers. Next, compute $\SSumX{S'}$ in $O(n\ell \log (n\ell)\log n)$ time using the algorithm of \lemref{stupid} (over the integers). Finally, compute the set $\left\{\sigma x \pmod m \:|\: \sigma\in \SSumX{S'}\right\} = \SSumX{S}$ in linear time.
\end{proof}

\begin{lemma}
    \lemlab{cycliccombine}
    Let $S \subseteq \Z_m$ be a set of size $n$ covered by segments $\seg{x_1}{\ell}, \ldots, \seg{x_k}{\ell}$, formally $S\subseteq \bigcup_{i=1}^k\seg{x_i}{\ell}$, then the set $\SSumX{S}$ can be computed in $O(k m\log m + n\ell \log (n\ell) \log n)$ time. 
\end{lemma}
\begin{proof}
Partition, in $O(kn)$ time, the elements of $S$ into $k$ sets $S_1,\ldots,S_k$, such that $S_i\subseteq \seg{x_i}{\ell}$, for $i \in \IntRange{k}$. Next, compute the subset sums $T_i = \SSumX{S_i}$ using the algorithm of \lemref{stupidcyclic}, for $i \in \IntRange{k}$. Then, compute $T_1 \oplus T_2 \oplus \ldots \oplus T_k = \SSumX{S}$, by $k-1$ applications of \lemref{o:plus}. The resulting running time is $O\bigl((k-1) m\log m + \sum_{i} |S_i|\ell \log (|S_i|\ell) \log |S_i|\bigr) = O(km\log m + n\ell \log (n\ell) \log n)$.
\end{proof}

\subsection{Covering a subset of $\U{\Z_m}$ by segments.}
Somewhat surprisingly, one can always find a short but ``heavy'' segment.

\begin{lemma}
    \lemlab{largecover}
    Let $S\subseteq U= \U{\Z_m}$, there exists a constant $c$, for any $\ell$ such that $c 2^{\frac{\ln m}{\ln \ln m}}\leq \ell\leq m$ there exists an element $x\in U$ such that $|\seg{x}{\ell} \cap S|=\Omega\left(\frac{\ell}{m}\,|S|\right)$.
\end{lemma}
\begin{proof}
Fix a $\beta\in U$. For $i\in U\cap \IntRange{\ell}$ consider the modular equation $ix \equiv \beta \pmod m$, this equation has a unique solution $x\in U$ -- here we are using the property that $i$ and $\beta$ are coprime to $m$. Let $\alpha=|U| / 2m$.
Let $\omega(m)$ be the number of distinct prime factors of $m$, and $\theta(m)=2^{\omega(m)}$ be the number of distinct square-free divisors of $m$. Then $\theta(m)\leq c 2^{\frac{\ln m}{\ln \ln m}}<\alpha\ell$ \cite{MR736719}.
There are at least $2\alpha\ell - \theta(m) \geq \alpha \ell$ elements in $U\cap \IntRange{\ell}$ \cite[Equation (1.4)]{Suryanarayana1974}.

Hence, when $\beta\in U$ is fixed, the number of values of $x$ such that $\beta\in \seg{x}{\ell}$ is at least $\alpha \ell$. Namely, every element of $S\subseteq U$ is covered by at least $\alpha \ell$ segments $\{\seg{x}{\ell} \:|\: x\in U\}$. As such, for a random $x\in U$ the expected number of elements of $S$ that are contained in $\seg{x}{\ell}$ is $\left(|S|\,\alpha \ell \right)/|U| = \frac{\ell}{2m} \, |S|$. Therefore, there must be a choice of $x$ such that $|\seg{x}{\ell} \cap S|$ is larger than the average, implying the claim.
\end{proof}

One can always find a small number of \segments~of length $\ell$ that contain all the elements of $\U{\Z_m}$.

\begin{lemma}
    \lemlab{smallcover}
    Let $S\subseteq \U{\Z_m}$ of size $n$, then for any $\ell$ such that $\ell\geq m^{1/2}$ there is a collection $\mathcal{L}$ of $O(\frac{m}{\ell}\ln n)$ \segments, each of length $\ell$, such that $S \subseteq \bigcup_{x\in \mathcal{L}} \seg{x}{\ell}$. Furthermore, such a cover can be computed in $O\bigl((n+\log m)\,\ell\bigr)$ time.
\end{lemma}
\begin{proof}
Consider the set system defined by the ground set $\Z_m$ and the sets $\{\seg{x}{\ell} \:|\: x\in \U{\Z_m}\}$. Next, consider the standard greedy set cover algorithm \cite{Johnson:1973:AAC:800125.804034,STEIN1974391,Lovasz:1975:ROI:2625515.2625652}: Pick a \segment~$\seg{x}{\ell}$ such that $|\seg{x}{\ell}\cap S|$ is maximized, remove all elements of $S$ covered by $\seg{x}{\ell}$, add $\seg{x}{\ell}$ to the cover, and repeat. By \lemref{largecover}, there is a choice of $x$ such that the segment $\seg{x}{\ell}$ contains at least a $c\ell/m$ fraction of $S$, for some constant $c$. After $m/c\ell$ iterations of this process, there will be at most $\left(1- c\ell/m\right)^{m/c\ell} n \leq n/e$ elements remaining. As such, after $O(\frac{m}{\ell}\ln n)$ iterations the original set $S$ is covered.

To implement this efficiently, in the preprocessing stage compute the modular inverses of every element in $\IntRange{\ell}$ using the extended Euclidean algorithm, in $O(\ell \log m)$ time \cite[Section 31.2]{cormen2014introduction}. Then, for every $b\in S$ and every $i \in \IntRange{\ell}$, find the unique $x$ (if it exists) such that $ix\equiv b \pmod m$, using the inverse $i^{-1}$ in $O(1)$ time. This indicates that $b$ is in $\seg{x}{\ell}\cap S$. Now, the algorithm computes $\seg{x}{\ell}\cap S$, for all $x$, in time $O(n\ell + \ell\log m)$. Next, feed the sets $\seg{x}{\ell}\cap S$, for all $x$, to a linear time greedy set cover algorithm and return the desired segments in $O(n\ell)$ time \cite[Section 35.3]{cormen2014introduction}. The total running time is $O\bigl((n+\log m)\,\ell\bigr)$.
\end{proof}

\subsection{Subset sums when all numbers are coprime to $m$.}
\label{setofunits}

\begin{lemma}
    \lemlab{unitcyclic}
    Let $S\subseteq \U{\Z_m}$ be a set of size $n$. Computing the set of all subset sums $\SSumX{S}$ takes $O\left(\min\left\{\sqrt{n}m, m^{5/4}\right\}\log m \log n\right)$ time. 
\end{lemma}
\begin{proof}
If $|S| \geq 2\sqrt{m}$, then $\SSumX{S}=\Z_m$ \cite[Theorem~1.1]{Hamidoune20081279}. As such, the case where $n = |S| \geq 2\sqrt{m}$ is immediate.

For the case that $n < 2\sqrt{m}$ we do the following. Apply the algorithm of \lemref{smallcover} for $\ell = m/\sqrt{n} \geq m^{1/2}$. This results in a cover of $S$ by $O(\frac{m}{\ell} \log n)$ segments (each of length $\ell$), which takes $O\bigl((n+\log m)\,\ell\bigr)= O(\sqrt{n}m \log m)$ time. Next, apply the algorithm of \lemref{cycliccombine} to compute $\SSumX{S}$ in $O(n\ell \log (n\ell)\log n)=O(\sqrt{n}m\log m\log n)$ time. Since, $n=O(\sqrt{m})$ this running time is $O\!\left(\min\left\{\sqrt{n}m, m^{5/4}\right\}\log m \log n\right)$.
\end{proof}

\subsection{The algorithm: Input is a subset of $\Z_m$.}
In this section, we show how to tackle the general case when $S$ is a subset of $\Z_m$. 

\subsubsection{Algorithm.}
\label{generalcyclicalgo}
The input instance is a triple $(\Gamma, \mu, \tau)$, where $\Gamma$ is a set, $\mu$ its \emph{modulus} and $\tau$ an auxiliary parameter. For such an instance $(\Gamma, \mu, \tau)$ the algorithm computes the set of all subset sums of $\Gamma \mbox{ modulo } \mu$. The initial instance is $(S, m, m)$. 

Let $q$ be the smallest prime factor of $\tau$, referred to as \emph{pivot}. Partition $\Gamma$ into the two sets: 
\[ \Gamma/q = \bigl\{s/q \:\big|\: s\in \Gamma \mbox{ and } q\divides s \bigr\} \mbox{ and } \Gamma\%q = \bigl\{s \in \Gamma \:\big|\: q\notdivides s \bigr\}\:. \]
Recursively compute the (partial) subset sums $\SSumX{\Gamma/q}$ and $\SSumX{\Gamma\%q}$, of the instances $(\Gamma/q, \mu/q, \tau/q )$ and $(\Gamma\%q, \mu, \tau/q )$, respectively. Then compute the set of all subset sums $\SSumX{\Gamma}= \bigl\{ qx \:\big|\: x\in \SSumX{\Gamma/q}\bigr\}\oplus \SSumX{\Gamma\%q}$ by combining them together using \lemref{o:plus}. At the bottom of the recursion, when $\tau = 1$, for each set compute its subset sums, using the algorithm of \lemref{unitcyclic}.

\subsubsection{Handling multiplicities.} 
During the execution of the algorithm there is a natural tree formed by the recursion. Consider an instance $(\Gamma, \mu, \tau)$ such that the pivot $q$ divides $\tau$ (and $\mu$) with multiplicity $r$. The top level recursion would generate instances with sets $\Gamma/q$ and $\Gamma\%q$. In the next level, $\Gamma/q$ is partitioned into $\Gamma/q^2$ and $(\Gamma/q)\% q$. On the other side of the recursion $\Gamma\%q$ gets partitioned (naively) into $(\Gamma\%q)/q$ (which is an empty set) and $(\Gamma\%q)\%q = \Gamma\%q$. As such, this is a superfluous step and can be skipped. Hence, compressing the $r$ levels of the recursion for this instance results in $r+1$ instances:
\[ \Gamma\%q,\, (\Gamma/q)\%q,\, \ldots,\, (\Gamma/q^{r-1})\%q,\, \Gamma/q^r \:. \]
The total size of these sets is equal to the size of $\Gamma$. In particular,  compress this subtree into a single level of recursion with the original call having $r+1$ children. At each such level of the tree label the edges by $0, 1, 2, \ldots, r$, based on the multiplicity of the divisor of the resulting (node) instance (i.e., an edge between instance sets $\Gamma$ and $(\Gamma/q^2)\%q$ would be labeled by ``2'').

\subsubsection{Analysis.} 
The recursion tree formed by the execution of the algorithm has a level for each of the $k=O(\log m / \log \log m)$ distinct prime factors of $m$ \cite{MR736719} -- assume the root level is the $0$th level. 

\begin{lemma}
    \lemlab{cyclicAnalysis}
    Consider running the algorithm on input $(S,m,m)$. Then the values of the moduli at the leaves of the recursion tree are \emph{unique}, and are precisely the divisors of $m$.
\end{lemma}
\begin{proof}
Let $m=\prod_{i=1}^k q_i^{r_i}$ be the prime factorization of $m$, where $q_i<q_{i+1}$ for all $1\leq i<k$. Then every vector $\bm{x}=(x_1,\ldots,x_k)$, with $0\leq x_i\leq r_i$, defines a path from the root to a leaf of modulus $m/\prod_{i=1}^k q_i^{x_i}$ in the natural way: Starting at the root, at each level of the tree follow the edge labeled $x_i$. If for two vectors $\bm{x}$ and $\bm{y}$ there is an $i \in \IntRange{k}$ such that $x_i \neq y_i$, then the two paths they define will be different (starting at the $i$th level). And, by the unique factorization of integers, the values of the moduli at the two leaves will also be different. Finally, note that every divisor of $m$, $\prod_{i=1}^k q_i^{\rho_i} \mbox{ with } 0\leq\rho_i\leq r_i$, occurs as a modulus of a leaf, and can be reached by following the path $(r_1-\rho_1, \ldots, r_k-\rho_k)$ down the tree.~
\end{proof}

\begin{theorem}
     Let $S\subseteq \Z_m$ be a set of size $n$. Computing the set of all subset sums $\SSumX{S}$ takes $O\bigl(\min\left\{\sqrt{n}m,m^{5/4}\right\}\log^2 m\bigr)$ time. 
\end{theorem}
\begin{proof}
The algorithm is described in Section \ref{generalcyclicalgo}, when the input is $(S,m,m)$. We break down the running time analysis into two parts: The running time at the leaves, and the running time at internal nodes. 

Let $\delta$ be the number of leaves of the recursion tree. Arrange them so the modulus of the $i$th leaf, $\mu_i$, is the $i$th largest divisor of $m$. Note that $\mu_i$ is at most $m/i$, for all $i \in \IntRange{\delta}$. Using \lemref{unitcyclic}, the running time is bounded by 
\begin{align*}
&O\!\left(\sum_{i=1}^{\delta} \min\left\{\sqrt{n_i}\,\mu_i,\mu_i^{5/4}\right\}\, \log n_i \log \mu_i\right) = O\!\left( \log m \log n \sum_{i=1}^{\delta} \min\left\{\sqrt{n_i}\,\frac{m}{i},\left(\frac{m}{i}\right)^{5/4}\right\}\right)\:.
\end{align*}
Using Cauchy-Schwartz, the first sum of the $\min$ is bounded by  
\[ m\sum_{i=1}^{\delta} \frac{\sqrt{n_i}}{i} \leq m\sqrt{\left(\sum_{i=1}^{\delta} \left(\sqrt{n_i}\right)^2 \right)\!\!\left(\sum_{i=1}^{\delta} \frac{1}{i^2}\right)} = O\!\left(\sqrt{n}m\right)\,, \]
and the second by $O(m^{5/4})$. Putting it all together, the total work done at the leaves is $O\left(\min\bigl\{\sqrt{n}m,m^{5/4}\bigr\}\log m \log n\right)$.

Next, consider an internal node of modulus $\mu$, pivot $q$ and $r+1$ children. The algorithm combines these instances, by applying $r$ times \lemref{o:plus}. The total running time necessary for this process is described next. As the moduli of the instances decrease geometrically, pair up the two smallest instances, combine them together, and in turn combine the result with the next (third) smallest instance, and so on. This yields a running time of
\[ O\left(\sum_{i=1}^r \frac{\mu}{q^i}\,\log\frac{\mu}{q^i}\right) = O(\mu \log \mu)\:. \]
At the leaf level, by \lemref{cyclicAnalysis}, the sum of the moduli $\sum_{i=1}^\delta \mu_i$ equals to $\sigma_1(m)$, and it is known that $\sigma_1(m)=O(m\log \log m)$ \cite[Theorem 323]{opac-b1099316}. As such, the sum of the moduli of all internal nodes is bounded by $O(k m \log \log m)=O(m\log m)$, as the sum of each level is bounded by the sum at the leaf level, and there are $k$ levels. As each internal node, with modulus $\mu$, takes $O(\mu\log \mu)$ time and $x\log x$ is a convex function, the total running time spent on all internal nodes is $O\bigl(m\log m \log (m \log m)\bigr) = O(m\log^2 m)$.

Aggregating everything together, the complete running time of the algorithm is bounded by $O\bigl(\min\left\{\sqrt{n}m,m^{5/4}\right\}\log^2 m\bigr)$, implying the theorem.
\end{proof}

The results of this section, along with the analysis of the recursion tree above, conclude the following corollary on covering $\Z_m$ with a small number of segments. The result is useful for error correction codes, and improves the recent bound of Chen et al. by a factor of $\sqrt{\ell}$ \cite{Chen2014}.

\begin{corollary}
    There exist a constant $c$, for all $\ell$ such that $c 2^{\frac{\ln m}{\ln \ln m}}\leq \ell \leq m$, 
    one can cover $\Z_m$ with $O\left((\sigma_1(m) \ln m)/\ell\right) + \sigma_0(m)$ segments of length $\ell$. Furthermore, such a cover can be computed in $O(m\ell)$ time. 
\end{corollary}
\begin{proof}
Let $S_{m/d}=\{ x/(m/d) \:|\: x\in \Z_m \mbox{ and } \gcd(x,m) = m/d \}$, for all $d\divides m$. Note that $S_{m/d} = \U{\Z_d}$, hence by \lemref{smallcover}, each $S_{m/d}$ has a cover of $O((d\ln d)/\ell)$ segments. Next, ``lift'' the segments of each set $S_{m/d}$ back up to $\Z_m$ (by multiplying by $m/d$) forming a cover of $\Z_m$. The number of segments in the final cover is bounded by 
\begin{align*}\sum_{\substack{d|m\\ \ell\leq d}} O\!\left(\frac{d}{\ell}\ln m\right) + \sum_{\substack{d|m\\ \ell>d}} 1 = O\!\left(\frac{\sigma_1(m) \ln m}{\ell}\right)+\sigma_0(m)\:.
\end{align*}
The time to cover each $S_{m/d}$, by \lemref{smallcover}, is $O\bigl((n+\log m)\,\ell\bigr)=O\bigl(\left(\varphi(d)+\log d\right)\ell\bigr)$, since there are $\varphi(d)$ elements in $S_{m/d}$, and $S_{m/d}\subseteq \Z_d$. Also, $\varphi(d)$ dominates $\log d$, as $O\bigl(\varphi(d)\bigr)= \Omega(d/\log \log d)$ \cite[Theorem 328]{opac-b1099316}, therefore the running time simplifies to $O\bigl(\varphi(d)\ell\bigr)$. Summing over all $S_{m/d}$ we have
\[ \sum_{d|m} O\bigl(\varphi(d)\ell\bigr) = O\!\left(\ell \sum_{d|m} \varphi(d)\right) = O(m\ell)\;, \]
since $\sum_{d|m} \varphi(d) = m$ \cite[Sec 16.2]{opac-b1099316}, implying the corollary.
\end{proof}

If $\ell<c 2^{\frac{\ln m}{\ln \ln m}}$, then $\ell=m^{o(1)}$. The corollary above then shows that for all $\ell$, there is a cover of $\Z_m$ with $m^{1+o(1)}/\ell$ segments. 

\section{Recovering the solution}
\label{sec:recover}

Given sets $X$ and $Y$, a number $x$ is a \emphi{witness} for $i\in X\oplus Y$, if $x\in X$ and $i-x\in Y$. A function $w : X\oplus Y\to X$ is a \emphi{witness function}, if $w(i)$ is a witness of $i$. 

If one can find a witness function for each $X\oplus Y$ computation of the algorithm, then we can traceback the recursion tree and reconstruct the subset that sums up to $t$ in $O(n)$ time. The problem of finding a witness function quickly can be reduced to the \emph{reconstruction problem} defined next.
\subsection{Reduction to the reconstruction problem.}

In the reconstruction problem, there are hidden sets $S_1,\ldots,S_n \subseteq \IntRange{m}$ and we have two oracles $\textsc{Size}$ and $\textsc{Sum}$ that take as input a query set $Q$.
\begin{itemize}
\item $\textsc{Size}(Q)$ returns the size of each intersection: \[ \bigl(|S_1\cap Q|,|S_2\cap Q|,\ldots,|S_n\cap Q|\bigr) \]
\item $\textsc{Sum}(Q)$ returns the sum of elements in each intersection: \[ \left(\sum_{s\in S_1\cap Q} s,\sum_{s\in S_2\cap Q} s,\ldots,\sum_{s\in S_n\cap Q} s\right) \]
\end{itemize}
The reconstruction problem asks to find $n$ values $x_1,\ldots,x_n$ such that for all $i$, if $S_i$ is non-empty, $x_i\in S_i$. Let $f$ be the running time of calling the oracles, and assume $f=\Omega(m+n)$, then is it known that one can find $x_1,\ldots,x_n$ in $O(f \log n \polylog m)$ time \cite{Aumann:2011:FWP:1921659.1921670}.

If $X,Y\subseteq \IntRange{u}$, finding the witness of $X\oplus Y$ is just a reconstruction problem. Here the hidden sets are $W_0,\ldots,W_{2u}\subseteq \IntRange{2u}$, where $W_i = \{x \:|\: x+y = i \mbox{ and } x\in X, y\in Y\}$ is the set of witnesses of $i$. Next, define the polynomials $\chi_Q(x) = \sum_{i \in Q} x^i$ and $I_Q(x) = \sum_{i \in Q} i x^i$. The coefficient for $x^i$ in $\chi_Q\chi_Y$ is $|W_i\cap Q|$ and in $I_Q\chi_Y$ is $\sum_{s\in W_i\cap Q} s$, which are precisely the $i$th coordinate of $\textsc{Size}(Q)$ and $\textsc{Sum}(Q)$, respectively. Hence, the oracles can be implemented using polynomial multiplication, in $\widetilde{O}(u)$ time per call. This yields an $\widetilde{O}(u)$ time deterministic algorithm to compute $X\oplus Y$ \emph{with} its witness function. 

Hence, with a polylogarithmic slowdown, we can find a witness function every time we perform a $\oplus$ operation, thus, effectively, maintaining which subsets sum up to which sum.

\section{Applications and extensions}
\label{applications}

\label{sec:applications}
Since every algorithm that uses subset sum as a subroutine can benefit from the new algorithm, we only highlight certain selected applications and some interesting extensions. Most of these applications are derived directly from the divide-and-conquer approach.

\subsection{Bottleneck graph partition.}

\label{bottleneckgraphpartition}
Let $G=(V,E)$ be a graph with $n$ vertices $m$ edges and let $w:E \rightarrow \R^+$ be a weight function on the edges. The \emph{bottleneck graph partition} problem is to split the vertices into two equal-sized sets such that the value of the bottleneck (maximum-weight) edge, over all edges across the cut, is minimized. This is the simplest example of a graph partition problem with cardinality constraints. The standard divide-and-conquer algorithm reduces this problem to solving $O(\log n)$ subset sum problems: Pick a weight, delete all edges with smaller weight and decide if there exists an arrangement of components that satisfy the size requirement \cite{NET:NET6}. The integers being summed are the various sizes of the components, the target value is $n/2$, and the sum of all inputs is $n$. Previously, using the $O(\sigma^{3/2})$ algorithm by Klinz and Woeginger, the best known running time was $O(m+n^{3/2}\log n)$ \cite{NET:NET5}. Using \thmref{stupid2}, this is improved to $O(m)+\widetilde{O}(n)$ time.

\subsection{All subset sums with cardinality information.}
Let $S = \{ s_1, s_2, \ldots, s_n\}$. Define $\bm{\sum}^{\leq n}_{\leq u} \pth {S}$ to be the set of pairs $(i,j)$, such that $(i,j) \in \bm{\sum}^{\leq n}_{\leq u} \pth {S}$ if and only if $i\leq u, j\leq n$ and there exists a subset of size $j$ in $S$ that sums up to $i$. We are interested in computing the set $\bm{\sum}^{\leq n}_{\leq u} \pth {S}$.

We are only aware of a folklore dynamic programming algorithm for this problem that runs in $O(n^2u)$ time. We include it here for completion. Let $D[i,j,k]$ be \texttt{true} if and only if there exists a subset of size $j$ that sums to $i$ using the first $k$ elements. The recursive relation is 
\begin{align*}
D[i,j,k] = \begin{cases}
\texttt{true} & \text{if }i=j=k=0\\
\texttt{false} & \text{if } i>j=k=0 \\
\!\!\!\begin{array}{l}
D[i,j,k-1]\vee D[i-s_k,j-1,k -1] 
\end{array} & \text{otherwise}
\end{cases}
\end{align*}
where we want to compute $D[i,j,n]$ for all $i\leq u$ and $j\leq n$. In the following we show how to do (significantly) better.


\begin{theorem}
\thmlab{subsetcard}
    Let $S \subseteq \IntRange{u}$ be a set of size $n$, then one can compute the set $\bm{\sum}^{\leq n}_{\leq u} \pth {S}$ in $O\bigl(nu \log (nu) \log n\bigr)$ time.
\end{theorem}
\begin{proof}
Partition $S$ into two (roughly) equally sized sets $S_1$ and $S_2$. Find $\bm{\sum}^{\leq n/2}_{\leq u} \pth {S_1}$ and $\bm{\sum}^{\leq n/2}_{\leq u} \pth {S_2}$ recursively, and combine them using \lemref{o:plus2}, in $O\bigl(nu \log (nu)\bigr)$ time. The final running time is then given by \obsref{rec}.
\end{proof}

\subsection{Counting and power index.}
Here we show that the standard divide-and-conquer algorithm can also answer the counting version of all subset sums. Namely, computing the function $N_{u,S}(x)$: the number of subsets of $S$ that sum up to $x$, where $x\leq u$. 

For two functions $f,g:X\to Y$, define $f \odot g:X\to Y$ to be 
\[ (f\odot g)(x) = \sum_{t\in X} f(x)g(x-t) \]

\begin{corollary}
    \lemlab{o:plusfunc}%
    Given two functions $f,g: \IY{0}{u} \to \N$ such that $f(x),g(x)\leq b$ for all $x$, one can compute $f \odot g$ in $O( u \log u \log b)$ time. 
\end{corollary}
\begin{proof} This is an immediate extension of \lemref{o:plus2} using the fact that multiplication of two degree $u$ polynomials, with coefficient size at most $b$, takes $O(u\log u\log b)$ time \cite{Schonhage1982}. 
\end{proof}

\begin{theorem}
    \thmlab{thm:counting}
    Let $S$ be a set of $n$ positive integers. One can compute the function $N_{u,S}$ in $O(nu \log u \log n)$ time.
\end{theorem}
\begin{proof}
Partition $S$ into two (roughly) equally sized sets $S_1$ and $S_2$. Compute $N_{u,S_1}$ and $N_{u,S_2}$ recursively, and combine them into $N_{u,S} = N_{u,S_1}\odot N_{u,S_1}$ using \lemref{o:plusfunc}, in $O(u \log u \log 2^n) = O(n u \log u)$ time. The final running time is then given by \obsref{rec}.
\end{proof}

\subsubsection{Power indices.}
The \emph{Banzhaf index} of a set $S$ of $n$ voters with cutoff $u$ can be recovered from $N_{u,S}$ in linear time. The \thmref{thm:counting} yields an algorithm for computing the Banzhaf index in $\widetilde{O}(nu)$ time. Previous dynamic programming algorithms take $O(nu)$ arithmetic operations, which translates to $O(n^2u)$ running time \cite{Uno2012}. Similar speed-ups (of, roughly, a factor $n$) can be obtained for the \emph{Shapley-Shubik index}.


\section*{Acknowledgments}

We would like to thank Sariel Har-Peled for his invaluable help in the editing of this paper as well as for various suggestions on improving the presentation of the results. We would also like to thank Jeff Erickson and Kent Quanrud for their insightful comments and feedback. We would like to thank Igor Shparlinski and Arne Winterhof for pointing out a problem in a proof. Finally, we also like to thank the anonymous reviewers for their helpful and meaningful comments on the different versions of this manuscript.

 
\bibliographystyle{plain} 
\bibliography{subset}

\newpage

\end{document}